\documentclass[twocolumn,doublespace,showkeys,showpacs,superscriptaddress]{IEEEtran}

\usepackage{graphicx,epic,eepic,epsfig,amsmath,latexsym,amssymb,verbatim,subfigure,color}
\usepackage{theorem}

\def\duzomniejsze{<\kern-.7mm<}
\def\duzowieksze{>\kern-.7mm>}

\def\textbf#1{{\bf #1}}
\def\beq{\begin{equation}}
\def\eeq{\end{equation}}
\def\be{\begin{equation}}
\def\ee{\end{equation}}
\def\ben{\begin{eqnarray}}
\def\een{\end{eqnarray}}
\def\beqa{\begin{eqnarray}}
\def\eeqa{\end{eqnarray}}
\def\eea{\end{array}}
\def\bea{\begin{array}}
\newcommand{\bei}{\begin{itemize}}
\newcommand{\eei}{\end{itemize}}
\newcommand{\bee}{\begin{enumerate}}
\newcommand{\eee}{\end{enumerate}}

\def\bare{{\bar{E}}}
\def\bcal{{\cal B}}
\def\pcal{{\cal P}}
\def\rcal{{\cal R}}
\def\hcal{{\cal H}}
\def\tr{{\rm Tr}}
\def\id{{\rm I}}

\def\>{\rangle}
\def\<{\langle}
\def\blacksquare{\vrule height 4pt width 3pt depth2pt}

\def\ot{\otimes}

\def\dt#1{{{\kern -.0mm\rm d}}#1\,}

\def\aaa{{A_1\ldots A_m}}
\def\aaacolon{{A_1:\ldots :A_m}}

\def\eqsq{E_{sq}^{q}}
\def\ecsq{E_{sq}^{c}}
\def\Sn{S_m}

\def\ep{{\epsilon}}

\def\ind{{i_1\ldots i_m}}
\def\systems{{(m)AA'}}
\def\psystems{{(m)A'}}
\def\csystems{{(m)A}}
\def\tsystems{{(m)\tilde{A}}}
\def\tsystemse{{\tsystems E}}

\def\initrho{\rho_{\tsystems}}
\def\initrhoe{\rho_{\tsystemse}}

\def\tmhilbert{{\hcal_{(m)\tilde{A}}}}
\def\pmhilbert{{\hcal_{(m)A'}}}
\def\mhilbert{{\hcal_{(m)A}}}
\def\tmnhilbert{{\hcal_{(m)\tilde{A}}^{(n)}}}

\def\tmhilberte{{\hcal_{(m)\tilde{A} E}}}

\def\pdit{pdit}

\def\proof{{\bf Proof.\ }}

\def\mnchie{{\chi_{(m)A E}^n}}

\def\mchi{{\chi_{(m)A}}}
\def\mnchiep{{\chi_{(m)A' E}^n}}
\def\mdist{{P_{(m)A}}}

\newtheorem{lemma}{Lemma}
\newtheorem{corrolary}{Corrolary}
\newtheorem{theorem}{Theorem}
\newtheorem{proposition}{Proposition}
\newtheorem{definition}{Definition}
\newtheorem{remark}{Remark}
\newtheorem{example}{Example}
\newtheorem{Obs}{Observation}

\begin{document}

\title{Squashed entanglement for multipartite states
and entanglement measures based on the mixed convex roof}
\author{Dong Yang, Karol Horodecki, Micha\l{} Horodecki, Pawe\l{} Horodecki, Jonathan Oppenheim, Wei Song
\thanks{Dong Yang was at the Laboratory for Quantum Information, China Jiliang University, Hangzhou, Zhejiang 310018, China.(e-mail:dyang@cjlu.edu.cn)}
\thanks{Karol Horodecki was at the Faculty of Mathematics, Physics and Computer Science, University of Gda\'nsk, 80--952 Gda\'nsk, Poland.(e-mail:khorodec@inf.univ.gda.pl)}
\thanks{Micha\l{} Horodecki was at the Institute of Theoretical Physics and Astrophysics
 University of Gda\'nsk, 80--952 Gda\'nsk, Poland.(e-mail:fizmh@univ.gda.pl)}
\thanks{Pawe\l{} Horodecki was at the Faculty of Applied Physics and Mathematics, Technical University of Gda\'nsk, 80--952 Gda\'nsk, Poland.(e-mail:pawel@mifgate.pg.gda.pl)}
\thanks{Jonathan Oppenheim was at the Department of Applied Mathematics and Theoretical Physics, University of Cambridge, U.K.(e-mail:J.Oppenheim@damtp.cam.ac.uk)}
\thanks{Wei Song was at the Department of Modern Physics, University of Science and Technology of China, Hefei, Anhui 230026, China.(e-mail:wsong1@mail.ustc.edu.cn)}
\thanks{We acknowledge support from EU grants IP QAP IST-015848 and IP SCALA IST-015714. J.O. also acknowledges the Royal Society. K.H. acknowledges the support of the Foundation for Polish Science. D.Y. acknowledges the support from NSF of China (Grant No. 10805043).}
}

\maketitle
\date{\today}

\begin{abstract}
  New measures of multipartite entanglement are constructed based on two definitions of multipartite information and     different methods of optimizing over extensions of the states. One is a generalization of the squashed entanglement where one takes the mutual information of parties conditioned on the state's extension and takes the infimum over such extensions. Additivity of the multipartite squashed entanglement is proved for both versions of the multipartite information which turn out to be related. The second one is based on taking classical extensions. This scheme is generalized, which enables to construct measures of entanglement based on the {\it mixed convex roof} of a quantity, which in  contrast to the standard convex roof method involves optimization  over all decompositions of a density matrix rather than just the decompositions into pure states. As one of the possible applications of these results we prove that any multipartite monotone is an upper bound on the amount of multipartite distillable key. The findings are finally related to analogous results in classical key agreement.

{\keywords Squashed entanglement, c-squashed entanglement, Mixed convex roof, Multipartite distillable key}
\end{abstract}

\section{Introduction}
\label{int} Many bipartite entanglement measures are known, but few are known for multipartite states. Among the known bipartite entanglement measures, the {\it squashed entanglement} is distinct from other measures in its additivity and interesting construction \cite{Tucci2002-squashed,Winter-squashed-ent}. It is based on quantum conditioning. In general, we do not know how to condition upon a quantum register. However if a function is built out of entropies, then there is a more or less magic way of conditioning, which is based on the fact that, accidentally, conditioning classical entropies can be represented by subtracting the conditioning entropy from the joint entropy -- the entropy of the state on $A$ conditioned by the state on $B$ can be written formally as \beq S(A|B)=S(AB)-S(B). \label{eq:cond} \eeq

Quantum conditional entropy often takes negative values. However, squashed entanglement is built out of the quantum mutual information, which remains positive upon conditioning. The definition of squashed entanglement is inspired by the intrinsic information, a quantity used in classical cryptography.

Interestingly, no multipartite generalization of this measure was proposed so far. In this paper we propose two generalizations, based on two versions of the multiparty  mutual information. One of these mutual information generalizations was considered already in \cite{Lindblad73-mutual,Tato-index}, while the second one, a kind of dual to the latter,  was introduced in \cite{Cerf-secr-mono}. The authors considered both  generalizations of mutual information to multipartite setting as so-called secrecy monotones. Thus it is natural to use both to generalize squashed entanglement. We prove that our multipartite versions of squashed entanglement are additive.


In the definition of bipartite squashed entanglement \cite{Winter-squashed-ent}, one considers conditional mutual information over extensions of the original state $\rho_{AB}$ to a third system $E$. We do the same in the multipartite case. The conditioning here is {\it quantum}: the system $E$ is a quantum register that means the extension state is quantum-correlated between $AB$ and $E$. This kind of conditioning is applicable only to linear combinations of quantum entropies.

On the contrary, we know how to condition {\it any} function upon a {\it classical} register $E$.
Here a question arises: {\it for which functions, can  we obtain a measure of entanglement based on classical conditioning?} In this paper we also attempt to answer this question and show that if the function $f$ satisfies the following conditions \bei \item[i)] $f$  is invariant under local unitary transformations \item[ii)] $f$ does not increase under local measurement on average \eei then, via conditioning, it gives  rise to an entanglement measure.  The measure turns out to be of  similar construction to the convex roof  \cite{Uhlmann-roof}, where one optimizes over all pure decompositions of a given state. Here we instead consider decompositions into mixed states (see e.g. \cite{Christandl:2002,Synak05-asym,YangHHS2005-cost}), \be E(\rho)=\inf_{\sum_i p_i\rho_i=\rho} \sum_i p_i f(\rho_i). \ee Such a construction is called the {\it mixed convex roof}. This approach works in both bipartite as well as the multipartite scenario. If the function is the mutual information, we have then two measures, the original squashed entanglement ({\it q-squashed} entanglement) and the classical squashed ({\it c-squashed}) entanglement. It should be mentioned that for bipartite systems the latter measure was already considered in \cite{Tucci2002-squashed} as an interpolation between entanglement of formation and squashed entanglement. It was also independently proposed in \cite{Christandl:2002,NagelR2003-ent-mutual}, however, monotonicity of this measure was not proved. Within our approach, the c-squashed entanglement is a special case of mixed convex roof measures, hence it is monotonous.

We also discuss asymptotic continuity of our entanglement measures. In particular, multipartite squashed entanglement is shown to satisfy this condition, which comes from the Alicki-Fannes result \cite{Alicki-Fannes} (cf. \cite{Winter-squashed-ent}). On the other hand the measures  based on mixed convex roof inherit the asymptotic continuity from the original functions \cite{Synak05-asym}.

The rest of the paper is organized as follows. Section \ref{bis} briefly recalls the construction of bipartite squashed entanglement. Section \ref{cond} discusses two formulas of multipartite mutual information, then defines their conditional  versions and explores their properties. Section \ref{mse} presents two multipartite entanglement measures. One is the q-squashed entanglement that is based on conditioning on a quantum register, which is the generalization of bipartite squashed entanglement. We consider two versions of q-squashed entanglement, based on two different formulas of multipartite mutual information. We prove additivity for both versions. The other entanglement measure is the c-squashed entanglement that uses conditioning on a classical register. It can also be obtained  via the mixed convex roof. Asymptotic continuity is also discussed. In Section \ref{sec:upper-bound}, we show that any multipartite monotone that satisfies reasonable axioms is an upper bound on distillable key, providing q-squashed entanglement as an example. It is a direct generalization of a bipartite result given in \cite{EkertCHHOR2006-ABEkey}, based on the multipartite key distillation protocol given in \cite{Augusiak-multi}. In Section \ref{subsec:classical-ubounds} we consider the classical analogue of the latter fact, that has been already presented in \cite{Cerf-secr-mono}. We consider here more carefully the issue of asymptotic continuity and the presence of the eavesdropper.  Section \ref{sum} provides a summary.

\section{Bipartite squashed entanglement}
\label{bis} In this section, we briefly recall the definition and properties of the original squashed entanglement. It is based on a simple quantity - the quantum mutual information, \be I(A:B)=S(A) + S(B) - S(AB), \ee where $S(X)$ is the von Neumann entropy of system $X$. Then one defines the conditional mutual information as \be I(A:B|E)= S(AE) + S(BE) - S(ABE) - S(E), \ee by simple application of the definition Eq. (\ref{eq:cond}) to the mutual information. The entanglement measure is then obtained by ``squashing''.
\begin{definition}
For a given state $\rho_{AB}$, squashed entanglement is given by \be E_{sq}(\rho_{AB}) = {1\over 2} \inf I(A:B|E), \ee where the infimum is taken over all extensions $\rho_{ABE}$ of the original state.
\end{definition}
The factor $1/2$ is taken just for convenience and an extension is any joint state $\rho_{ABE}$ such that tracing out the state on $E$ leaves $\rho_{AB}$ on $AB$.

Squashed entanglement satisfies not only the basic properties of entanglement measures - monotonicity under local operations and classical communication (LOCC), but also other useful properties, notably additivity. It is also a good upper bound for distillable entanglement.

\section{Conditional multipartite mutual information}
\label{cond}

In this section, we discuss the mutual information for multipartite states and then study the properties of conditional multipartite mutual information.

There are at least two candidates for multipartite mutual information. One is more in the spirit of the Venn diagram representation of information-theoretic identities \cite{CoverThomas}, where the bipartite mutual information corresponds to the intersection of two sets. To extend this notion to the multipartite case, we would need to take the intersection of all sets. For example in the tripartite case this approach would give the following candidate for mutual information \be S(A)+S(B)+S(C) - S(AB) -S(AC) - S(BC) + S(ABC). \ee
This function can however be negative. The other approach, which produces a much simpler function, is to notice that mutual information is the relative entropy between the state and its marginal \be I(\rho_{AB})=S(\rho_{AB}||\rho_A\ot \rho_B), \ee where the relative entropy is defined as \be S(\rho||\sigma)=tr\rho(\log\rho-\log\sigma).\ee Thus multipartite mutual information would be \be I(A:B:C)=S(\rho_{ABC}|| \rho_A\ot \rho_B \ot \rho_C). \ee
\begin{figure}
\label{fig:venn} \centering
\includegraphics[scale=1.2]{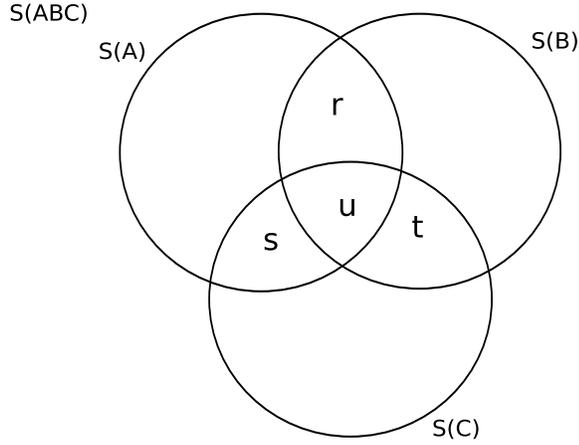}
\caption{A graphical representation of entropic quantities for the case of three parties. Here, the region that represents $I(A:B)$ is $r+u$.  $S_3=r+s+t+u$ and $I(A:B:C)=r+s+t+2u$}
\end{figure}
This turns out to be a very simple function considered by Lindblad \cite{Lindblad73-mutual} and used in \cite{Tato-index} to describe correlations within multipartite quantum systems:\be I(A_1:A_2:\ldots : A_m )=\sum_{i=1}^m S(A_i) - S(A_1 A_2 \ldots A_m). \ee In this paper we will adopt this function as one measure of multipartite mutual information. Since it can be represented as the relative entropy between a state and the product of its marginals, it is a monotone under local operations (because of the monotonicity of relative entropy \cite{Lindblad75-relative,Uhlmann-relative}).

This version of multipartite mutual information has yet another interesting feature \cite{Cerf-secr-mono}: it can be represented as a sum of bipartite mutual information\ben I(A_1:A_2:\ldots : A_m )=I(A_1:A_2) + I(A_3:A_1A_2)\nonumber\\ + I(A_4:A_1A_2 A_3) + \ldots I(A_m:A_1 A_2 \ldots A_{m-1}) \label{eq:multi-bi}. \een Put in a different way, $I$ satisfies the recurrence relation \ben I(A_1:A_2:\ldots : A_m)&=&I(A_1:A_2:\ldots : A_{m-1})\nonumber\\&&+ I(A_m:A_1\ldots A_{m-1}) \label{eq:recur-I}. \een So it is actually a very simple quantity, which is built as follows: first we consider the (original) mutual information between two chosen parties. Then we want to take into account the next party, so we take the mutual information between it and the previous two parties. In general, to take into account a new party means to add the mutual information between this party and all the other ones (treated jointly). Moreover this quantity is independent of the selected order in which we combine them. It is also obvious that the expression is invariant under permutation of the systems. The fact that multipartite mutual information can be decomposed into bipartite mutual information already indicates that most of the proofs for the original mutual information in the bipartite case can be applied to multipartite case without great modification.

There is yet another formula for mutual information in the multipartite case, as proposed in \cite{Cerf-secr-mono}, where
multipartite {\it secrecy monotones} were considered. It is of the form \be \Sn=\sum_{i=1}^m S(\rho_{A_1,\ldots, A_{i-1},A_{i+1},\ldots,A_m})-(m-1) S(\rho_{A_1\ldots A_m}). \label{eq:defsn} \ee It has an alternative expression in terms of conditional mutual information \cite{Cerf-secr-mono} (cf. Eq. \ref{eq:multi-bi}): \ben \Sn &=&I(A_1:A_2 \ldots  A_m) + I(A_2:A_3\ldots A_m|A_1) \nonumber\\
&&+ I(A_3:A_4\ldots A_m|A_1A_2)+ \ldots \nonumber\\
&&+ I(A_{m-1}:A_m|A_1\ldots A_{m-2}). \een It turns out that the two candidates for mutual information, $I$ and $S_m$ are dual to each other: \ben &&\Sn+ I(A_1:A_2:\ldots : A_m )\nonumber\\ &=& \sum_{i=1}^m I(A_i:A_1\ldots A_{i-1} A_{i+1} \ldots A_m) \label{eq:s-plus-i}, \een i.e. they add up to the sum of all bipartite mutual information between a single system and the rest. In the case of bipartite systems, both candidates are equal to the bipartite mutual information.

In the following section we will show that both mutual information formulas can be used to construct a good entanglement monotone.

Before we define the entanglement measure, let us discuss the problem of conditioning. In general, we do not know what it means to condition upon a quantum register. However it turns out that if a function is a linear combination of entropies, one can condition such a function using a mysterious prescription  for conditioning entropy: \be S(A|B)= S(AB) - S(B).\nonumber \ee This quantity can be negative, however it has a good interpretation \cite{SW-nature}. For linear combination of entropies, we apply this formula to each term. Such conditioning is well defined, in a sense that the result of conditioning does not depend on how we have represented the function in terms of entropies, e.g. a function $f(AB)=S(A) + S(AB) - 2 S(AB)$ has the same conditioned version as in the case it is written as $f(AB)=S(A) - S(AB)$. This way of conditioning is rather blind, and it does not tell us how to do it in the case of other functions.

Note that this conditioning reduces to the usual one when the register $E$ is classical. Another interesting feature is the following \be f(A_1, \ldots, A_m|EE') = f((A_1, \ldots, A_m|E)|E') \label{eq:cond-comp}, \ee for $f$ a linear combination of entropies and the right hand side has the following meaning: we first compute the function $f(A_1,\ldots, A_m|E)$. It is again linear combination of entropies, so one can apply to it conditioning again. Thus conditioning is {\it associative}.

Since conditioning does not depend on how we decompose the function into entropies, we get that the conditional mutual information is a sum of conditional bipartite mutual information. For example we have: \ben I(A_1:A_2: \ldots |E) =I(A_1:A_2|E)+I(A_3:A_1A_2|E) \nonumber\\
+ \ldots+I(A_m:A_1\ldots A_{m-1}|E). \label{eq:multi-bi-cond} \een

Let us examine properties of both mutual and conditional mutual information in the multipartite case. We will provide here two properties that were proved in \cite{Winter-squashed-ent} for the bipartite case. It would be very convenient to have a chain rule, however, we do not have it. In place of the chain rule we have two identities that still prove useful:

\begin{proposition}{\it The following
identities hold for multipartite mutual informations.}
\ben
&&I(X\aaacolon)=\nonumber\\
&&I(\aaacolon|X)+\sum_{i=2}^m I(X:A_i),
\label{eq:rule1}\\
&&I(X\aaacolon)=\nonumber\\
&&I(\aaacolon)+\sum_{i=2}^m I(X:A_i|A_1\ldots A_{i-1}),
\label{eq:rule2}\\
&&\Sn(X\aaacolon)=\nonumber\\
&&\Sn(\aaacolon|X) + I(X:A_2\ldots A_m),
\label{eq:rule3}\\
&&\Sn(X\aaacolon)=-I(A_2:\ldots:A_m)\nonumber\\
&&+\sum_{i=2}^{m}I(A_i:XA_1\ldots A_{i-1} A_{i+1} \ldots A_m). \label{eq:rule4} \een Moreover conditional mutual informations are additive.
\end{proposition}
{\bf Proof.} To get identities for $I$ one uses repeatedly the chain rule for bipartite mutual information that asserts $I(XY:Z)=I(X:Z)+I(Y:Z|X)$, and Eq.(\ref{eq:multi-bi}) and Eq.(\ref{eq:multi-bi-cond}). Eq. (\ref{eq:rule3}) follows directly from definition (\ref{eq:defsn}). The last identity (see \cite{Cerf-secr-mono}) follows from Eqs. (\ref{eq:recur-I}) and  (\ref{eq:s-plus-i}).

Additivity follows from the fact that conditional mutual informations are  linear combinations of entropies which are additive.

We now immediately get basic monotonicity properties of mutual information (both $I$ and $\Sn$), which are enough to ensure that ``squashing" will produce a function that  will be a monotone under LOCC.
\begin{proposition}
\label{prop:im-mon} The following statements hold. \bei \item[(i)] Mutual informations do not increase under local conditioning, \ben I(X\aaacolon) \geq I(\aaacolon|X),\nonumber\\ \Sn(X\aaacolon) \geq \Sn(\aaacolon|X). \label{eq:mon-c} \een \item[(ii)] Conditional mutual informations do not increase under local operations, \ben I(\aaacolon|E)\geq I(\Lambda(A_1):\ldots:A_m|E),\nonumber\\ \Sn(\aaacolon|E)\geq \Sn(\Lambda(A_1):\ldots:A_m|E), \een where $\Lambda(\cdot)$ is a local operation on subsystem $A_1$ (a trace-preserving completely positive map).\eei
\end{proposition}

{\bf Proof.} The statement (i) follows from Eqs.(\ref{eq:rule1}) (\ref{eq:rule3}), and the fact that bipartite mutual information is nonnegative. The statement (ii) comes from two facts: one is that any local operation can be realized in the unitary representation $\Lambda^Q(\rho_Q)=tr_R U_{RQ}(|0\>\<0|)_{R}\otimes\rho_{Q}U_{RQ}^{\dagger}$ and $I(\cdot|E)$ is invariant under local unitary operation and adding ancilla; the other fact is that \be I(X\aaacolon|E) \geq I(\aaacolon|E)\label{cond:tr}.\ee To get Eq.(\ref{cond:tr}), we apply conditioning to both sides of Eq.(\ref{eq:rule2}) and using associativity of conditioning obtain \ben I(X\aaacolon|E)=I(\aaacolon|E)\nonumber\\ +\sum_{i=2}^m I(X:A_i|A_1\ldots A_{i-1}E). \een Then, the required inequality follows from the fact that bipartite conditional mutual information is nonnegative. Just by conditioning  (\ref{eq:rule4}) we get the statement (ii) for $\Sn$. This completes the proof. \blacksquare

\section{Multipartite squashed entanglement}
\label{mse} In this section we will define two squashed measures and prove that they are monotones under LOCC. One is the q-squashed entanglement when the extension register is quantum. The other is the c-squashed entanglement when the extension register is classical.

\subsection{Monotonicity of entanglement measures under LOCC}
First we review what is meant by a monotone under LOCC.
\begin{definition}
A function $E$ is called a LOCC monotone if it satisfies the following condition: \be E(\rho_{AB})\geq \sum_ip_i E(\sigma_{AB}^i), \ee where $i$ are outcomes of the LOCC operation, $p_i$ are probabilities of outcomes, and $\sigma_{AB}^i$  is the state  given that outcome $i$ was obtained. \label{def:mon}
\end{definition}

We have the following proposition \cite{MH2004-mono}.
\begin{proposition}
\label{prop:locc1} A convex function $f$ does not increase under  LOCC if and only if

\noindent{\rm [LUI]} $f$ satisfies  local unitary invariance \be f(U_A\ot U_B \rho_{AB} U_A^\dagger\ot U_B^\dagger)= f(\rho_{AB}), \ee

\noindent {\rm [FLAGS]} $f$ satisfies \be f(\sum_ip_i \rho_{AB}^i \ot |i\>_X\<i|)=\sum_ip_i f( \rho_{AB}^i), \label{eq:flags} \ee for $X=A',B'$ where $|i\>$ are local, orthogonal ``flags''.
\end{proposition}

\subsection{Q-squashed entanglement}
\label{sec:mon} Now we define the multipartite q-squashed entanglement as follows
\begin{definition}
For the $m$ party state $\rho_{A_1,\ldots, A_m}$ \be \eqsq(\rho_{A_1,\ldots, A_m})=\inf I(A_1:A_2:\ldots :A_m|E), \ee where the infimum is taken over states $\sigma_{A_1,\ldots, A_m,E}$ that are extensions of $\rho_{A_1,\ldots, A_m}$, i.e. $\tr_E \sigma= \rho$.
\end{definition}
Alternatively, one can use $\Sn$. We will perform all proofs for $I$, however we will use only Proposition \ref{prop:im-mon}, as well as the fact that if the register $E$ is classical, then conditioning upon it regains its usual meaning. Thus all results are valid also for $S_m$, so that it gives rise to an independent definition of multipartite squashed entanglement.

The measures we propose reduce to twice the original squashed entanglement in the case of two parties. It resembles in some respect  relative entropy distance from the set of completely separable states, i.e. states of the form \be \sum_i p_i \rho^i_{A_1}\ot  \ldots\ot \rho^i_{A_m}. \ee Namely, it attributes nonzero entanglement to states that do not contain true multipartite entanglement, such as many singlet pairs. It is immediate to show that $\eqsq$ satisfies condition LUI. In the following we will show that it is convex, and satisfies FLAGS.

{\bf Convexity}. Consider two states $\rho_{\aaa }$ and $\sigma_{\aaa}$, and their arbitrary extensions $\rho_{\aaa E}$ and $\sigma_{\aaa E}$. We have \be \eqsq(p \rho_{\aaa} + (1-p) \sigma_{\aaa}) \leq I(\aaa|EE')_{\tau} \ee with \ben \tau_{\aaa EE'}=p \rho_{\aaa E}\ot |0\>_{E'}\<0| \nonumber\\+(1-p) \sigma_{\aaa E}\ot |1\>_{E'}\<1|. \een This is because the latter state is an extension of the state $p \rho + (1-p) \sigma$. However, since register $E'$ is classical, the conditioning is reduced to the usual one, so that we have \ben I(\aaa|EE')_{\tau} =p I(\aaa|E)_\rho \nonumber\\+ (1-p) I(\aaa|E)_\sigma, \een where we have used associativity of conditioning. Since we can choose arbitrary extensions of $\rho$ and $\sigma$, we get that right-hand-side can be arbitrary close to $p \eqsq(\rho) + (1-p) \eqsq(\sigma)$. This proves convexity.\ \blacksquare

{\bf FLAGS}. We decompose the condition into two inequalities. The inequality \be \eqsq(\sum_ip_i \rho_{\aaa}^i \ot |i\>_X\<i|)\leq \sum_ip_i \eqsq( \rho_{\aaa}^i) \ee follows from convexity (which we have just proved) and from invariance of $\eqsq$ under adding local pure ancilla. The latter is obvious, since any extension of a state $\rho_{\aaa}\ot |0\>_{A_1'}\<0|$ is of the form $\rho_{\aaa E} \ot |0\>_{A_1'}\<0|$ where $\rho_{\aaa E}$ is extension of $\rho_\aaa$. (Then it follows from additivity of entropy.)

The converse inequality \be \eqsq(\sum_ip_i \rho_{\aaa}^i \ot |i\>_X\<i|)\geq \sum_ip_i \eqsq( \rho_{\aaa}^i) \ee can be interpreted as non-increasing upon conditioning under the local classical register. We have \be \eqsq(\sum_ip_i \rho_{\aaa}^i \ot |i\>_X\<i|)=\inf I(X\aaacolon|E)_\tau, \ee where $\tau$ is the extension of the state $\sum_ip_i \rho_{\aaa}^i \ot |i\>_X\<i|$. Now we dephase register $X$. Since it is a local operation, according to Prop. \ref{prop:im-mon} the conditional mutual information can only go down. Thus we have \ben && I(X\aaacolon|E)_\tau \geq
I(X\aaacolon|E)_{\tilde \tau}\geq  \nonumber \\
&& \geq I(\aaacolon|EX)_{\tilde\tau}=\sum_ip_i I(\aaacolon|E)_{\tau_i} \nonumber\\
&& \geq \sum_ip_i\eqsq(\rho_\aaa), \een where $\tilde \tau$ is the dephased version of $\tau$ and $\tau_i$ is the dephased version of $\tau$, given $X=i$, which is still a valid extension of our state. The second inequality comes from Prop. \ref{prop:im-mon}. The equality comes from the fact that after dephasing the register is classical. \ \blacksquare

From the definition of $\eqsq$, we can see that no assumption on the classicality of register $E$ is made. So generally the register $E$ is quantum-correlated with system $\aaa$ in the extensions. What about the case when we allow only classical registers in the extensions i.e. the register and the system is classically correlated? Can it still be a monotone under LOCC? The answer is yes. This comes from the result that we know how to condition any function upon classical registers. In fact, we will show that mutual information is just a special case immediately obtained from a general theorem.

{\bf Remark:} In a recent paper \cite{Oppenheim-paradigm-E}, the bipartite squashed entanglement was given an operational interpretation in terms of communication rate in state merging. In \cite{Hayden-Multisquash}, the multipartite q-squashed entanglement on $I_m$ with the factor $1/2$ was developed independently and was related to the outer bound on the communication rate in multiparty distributed compression. To find a similar situation for the mutipartite q-squashed entanglement on $S_m$ is an open question.

\subsection{Additivity of q-squashed entanglement}
\label{sec:add} In this section we show that the two versions of multipartite q-squashed entanglement are additive on tensor product and subadditive in general. Namely it turns out that the main ideas behind the proof of additivity of bipartite squashed entanglement can be carried over  to multipartite case. We will show explicitly the proof for three parties. The generalization of n parties is immediate.

The main ingredient is the following lemma. {\lemma \label{lem:chain} \ben
I(AA':BB':CC'|E)= I(A:B:C|A'B'C'E)
\nonumber \\+I(A':B':C'|E)+I(A:B'C'|A'E)
\nonumber \\+I(B:A'C'|B'E)+I(C:A'B'|C'E),\\
{}\nonumber\\
S_3(AA':BB':CC'|E)= S_3(A:B:C|A'B'C'E)
\nonumber\\+S_3(A':B':C'|E)+I(AB:C'|A'B'E)
\nonumber\\+I(AC:B'|A'C'E)+I(BC:A'|B'C'E). \een }

{\it Proof. ~} \ben
&&I(AA':BB':CC'|E)-I(A:B:C|A'B'C'E)\nonumber\\
&&~~~~~~~~~~~~~~~~~~~~~~~~~~~-I(A':B':C'|E)\nonumber\\
&&=S(AA'E)+S(BB'E)+S(CC'E)\nonumber\\
&&+3S(A'B'C'E)-S(AA'B'C'E)-S(BA'B'C'E)\nonumber\\
&&-S(CA'B'C'E)-S(A'E)-S(B'E)-S(C'E)\nonumber\\
&&=S(AA'E)+S(A'B'C'E)-S(AA'B'C'E)\nonumber\\
&&-S(A'E)+S(BB'E)+S(A'B'C'E)\nonumber\\
&&-S(BA'B'C'E)-S(B'E)+S(CC'E)\nonumber\\
&&+S(A'B'C'E)-S(CA'B'C'E)-S(C'E).
\een

\ben
&&S_3(AA':BB':CC'|E)-S_3(A:B:C|A'B'C'E)\nonumber\\
&&~~~~~~~~~~~~~~~~~~~~~~~~~~~~~-S_3(A':B':C'|E)\nonumber\\
&&=S(AA'BB'E)+S(AA'CC'E)+S(BB'CC'E)\nonumber\\
&&~~~~~~~~~~~~~~~~~~~~~~~~~~~~~-S(ABA'B'C'E)\nonumber\\
&&-S(ACA'B'C'E)-S(BCA'B'C'E)+3S(A'B'C'E)\nonumber\\
&&-S(A'B'E)-S(A'C'E)-S(B'C'E)\nonumber\\
&&=I(AB:C'|A'B'E)+I(BC:A'|B'C'E)\nonumber\\
&&~~~~~~~~~~~~~~~~~~~~~~~~~~~+I(CA:B'|C'A'E).\blacksquare\een
Since conditional mutual information is positive, we obtain immediately

{\corrolary \label{cor:superadd} \ben
I(AA':BB':CC'|E)\ge I(A:B:C|A'B'C'E)\nonumber\\+I(A':B':C'|E),\\[2mm]
S_3(AA':BB':CC'|E)\ge S_3(A:B:C|A'B'C'E)\nonumber\\+S_3(A':B':C'|E). \een } The lemma \ref{lem:chain} allows also to prove: {\lemma \label{lem:add} {\bf Conditional $I$ and $S_3$ are additive.} For product states $\rho_{ABCE}\otimes \sigma_{A'B'C'E'}$ \ben
I(AA':BB':CC'|EE')= I(A:B:C|E)\nonumber\\+I(A':B':C'|E'),\\[2mm]
S_3(AA':BB':CC'|EE')= S_3(A:B:C|E)\nonumber\\+S_3(A':B':C'|E'). \een } {\it Proof.} From lemma \ref{lem:chain} we get \ben
I(AA':BB':CC'|EE')&=&I(A:B:C|A'B'C'EE')\nonumber\\+I(A':B':C'|EE')
&+&\underbrace{I(A:B'C'|A'EE')}_{=0}\nonumber\\+\underbrace{I(B:A'C'|B'EE')}_{=0}&+&\underbrace{I(C:A'B'|C'EE')}_{=0},\nonumber\\
=I(A:B:C|E)&+&I(A':B':C'|E'). \een
The last three terms vanish because the state is a tensor product: the primed systems are not correlated with the non-primed. Similarly:
\ben
S_3(AA':BB':CC'|EE')&=& S_3(A:B:C|A'B'C'EE')\nonumber\\
+S_3(A':B':C'|EE')&+&\underbrace{I(AB:C'|A'B'EE')}_{=0}      \nonumber\\
+\underbrace{I(AC:B'|A'C'EE')}_{=0}&+&\underbrace{I(BC:A'|B'C'EE')}_{=0}
\nonumber\\ = S_3(A:B:C|E)&+&S_3(A':B':C'|E'),
\een
which ends the proof.
An immediate conclusion from the  corollary \ref{cor:superadd} and  lemma \ref{lem:add}  is that $E_{sq}^q$ based on either of the mutual informations is additive. {\proposition For any $n$-partite states $\rho$ and $\sigma$, \be E_{sq}^q(\rho\otimes \sigma)=E_{sq}^q(\rho)+E_{sq}^q(\sigma). \ee } From Corollary \ref{cor:superadd} it follows that $E_{sq}^q$ is superadditive, \be E_{sq}^q(\rho_{12})\geq E_{sq}^q(\rho_1)+E_{sq}^q(\rho_2), \ee where the subscripts denote two $n$ partite systems.

\subsection{Entanglement measures via mixed convex roof}
In this subsection, we prove a theorem that provides a method of constructing entanglement measures via mixed convex roof.

\begin{definition} A mixed convex roof of a function $g$ is given by
\be E_g(\rho)=\inf \sum_i p_i g(\rho_i), \ee where infimum is taken over all ensembles $\{ p_i, \rho_i\}$ satisfying $\sum_i p_i \rho_i=\rho$.
\end{definition}

{\bf Remark:} In \cite{Synak05-asym} it was shown that if $g$ is continuous, then there exists a finite decomposition that realizes the infimum. Moreover it was shown, that if $g$ is asymptotically continuous, then so is its mixed convex roof (which is a generalization  of \cite{Nielsen-cont}).

\begin{theorem} For any continuous function $g$ which is a monotone on average under local operations,
its mixed convex roof is monotonic under LOCC.
\end{theorem}

Let us recall that to prove monotonicity of multipartite squashed entanglement we have used two features of mutual information proven in proposition \ref{prop:im-mon}: monotonicity under local operations, and monotonicity under local conditioning. Note that those two features, if we consider only conditioning upon classical register, are equivalent to monotonicity on average under local operations.

{\bf Proof.} First, it is immediate to see that  $E_g$ is convex. It remains to show that also $E_g$ is monotone under local measurements, on average \cite{Vidal-mon2000}. We will simply show that for any given measurement, if $g$ does not increase under this measurement for any state, $E_g$ does not either. Consider arbitrary state $\rho$ with optimal decomposition $\{p_i,\rho_i\}$ (if there does not  exist such ensemble, we can  take a nearly optimal one, and the proof still goes through). The measurement transforms the state $\rho$ into ensemble: \be \rho \to \{ q_k, \sigma_k\} \ee and the members $\rho_i$ of the optimal ensemble into ensembles $\{q_{k|i}, \sigma_k^{(i)}\}$. Thus $q_{k|i}$  is the probability of obtaining outcome $k$ provided the state $\rho_i$ was sent. The probability that $\rho_i$ was sent, if outcome $k$ was obtained  is  given by $q_{i|k}=p_i q_{k|i}/ q_k$. Note that \be \sum_i q_{i|k} \sigma_k^{(i)} = \sigma_k. \label{eq-ik} \ee Now we have \ben E_g(\rho)= \sum_i p_i g(\rho_i)\geq \sum_i p_i  \sum_k q_{k|i} g(\sigma_k^{(i)})\nonumber\\=\sum_k q_k \sum_iq_{i|k} g(\sigma_k^{(i)})\geq  \sum_k q_k E_g(\sigma_k). \een The first equality comes from the fact that the ensemble $\{p_i, \rho_i\}$ is optimal, whose existence is guaranteed by continuity of $g$ (see above remark). The first inequality comes from monotonicity of $g$ under the considered measurement. The last inequality, comes from the fact that $\{q_{i|k},\sigma_k^{(i)}\}$ is some ensemble of $\sigma_k$ (it follows from Eq.(\ref{eq-ik})). This completes the proof.\ \blacksquare

A new entanglement measure can be constructed via mixed convex roof if a function satisfying the condition of the theorem is found. An example illustrating the theorem
 is that one can take as $g$ the function $I_c$, defined as supremum of mutual
information of the joint probability distribution after Bob and Alice's measurement. That is \be I_c(\rho_{AB})=\sup_{A_{i}^{\dagger}A_i,B_{i}^{\dagger}B_i}H(p_A)+H(p_B)-H(p_{AB}), \ee where $p^{i}_{AB}=tr A_i\otimes B_{i}\rho_{AB}A_{i}^{\dagger}\otimes B_{i}^{\dagger}$. Another one is the c-squashed entanglement defined below.

\subsection{C-squashed entanglement}
Multipartite c-squashed entanglement is defined as follows.
\begin{definition}
For the $m$ party state $\rho_{A_1,\ldots, A_m}$ \be \ecsq(\rho_{A_1,\ldots, A_m})=\inf I(A_1:A_2:\ldots :A_m|E), \ee where infimum is taken over the extension states $\sigma_{A_1,\ldots, A_m,E}$ of the form $\sum p_i\rho_{\aaa}^{i}\otimes |i\>_{E}\<i|$.
\end{definition}

It is clear that $\ecsq$ is nothing but the entanglement measure obtained by mixed convex roof of multipartite quantum mutual information function. We have seen that multipartite quantum mutual information  satisfies the conditions of the theorem. This function  was considered in \cite{Christandl:2002} in the context of distillation of private key. Its bipartite version was also presented in \cite{Tucci2002-squashed,NagelR2003-ent-mutual}. By the theorem, we get that it is indeed a good entanglement measure. The measure is related to the q-squashed entanglement: it is obtained, if in the definition of ``squashing", one restricts to extensions such that the auxiliary system is classical.

Of course by replacing $I$ with $S_m$ we obtain another version of c-squashed entanglement. In bipartite case, it coincides with the previous one.

\subsection{Asymptotic continuity}
\label{sec:cont} In this subsection, we prove the asymptotic continuity of $\eqsq$ and $\ecsq$.
\begin{definition}
Let $f$ be a real-valued function $f$ defined on states $\rho_1,\rho_2$ acting on Hilbert space $C^d$ and $\epsilon=\|\rho_1-\rho_2\|_1$. Then a function is asymptotically continuous if it fulfils the following condition \be \forall_{\rho_1,\rho_2}|f(\rho_1)-f(\rho_2)|\le K_1\epsilon\log d+O(\epsilon), \ee where $K_1$ is constant and $O(\epsilon)$ is any function, which satisfies $\lim_{\epsilon\to 0}O(\epsilon)\to 0$ and depends only on $\epsilon$. (In particular, it does not depend on the dimension).
\end{definition}

\begin{proposition}
The multipartite q-squashed entanglement $\eqsq(\aaa)$ and c-squashed entanglement $\ecsq(\aaa)$ are asymptotically continuous.
\end{proposition}
{\it Proof.} The asymptotic continuity of $\eqsq(\aaa)$ comes from Eq.(\ref{eq:multi-bi-cond}) and Alicki-Fannes result \cite{Alicki-Fannes} that asserts $I(A:B|E)$ is asymptotically continuous and $d=d_A$. The asymptotic continuity of $\ecsq(\aaa)$ comes from the fact that the mixed convex roof of a function inherit the asymptotic continuity of that function \cite{Synak05-asym}.

\subsection{Lockability property}
\label{se:lockable} In this subsection, we will calculate squashed entanglement for the multipartite flower state and show that multipartite squashed entanglement is lockable. A quantity, usually some measure of correlations, is lockable if it can drop by an arbitrary amount when only a single qubit of a local system is lost . This property has previously been observed for many correlation measures such as accessible information, entanglement cost and logarithmic negativity \cite{DiVincenzo-locking,lock-ent,KoenigRBM-locking}. The bipartite squashed entanglement has been proved to be lockable in \cite{CW-lock}. Here we adapt the proof to the multipartite case. First we generalize the flower state to multipartite case.
\begin{definition}
A tripartite flower state $\rho_{AA'BB'CC'}$ is defined through its purification state $|\Phi\>_{AA'BB'CC'X}$ of the form, \ben |\Phi\>_{AA'BB'CC'X} &=& \nonumber \\{1\over{\sqrt{2d}}}\sum_{i=1,\ldots,d\atop j=0,1}|i\>_A|j\>_{A'}|i\>_B|j\>_{B'}|i\>_C|j\>_{C'}U_j|i\>_X,
\een where $U_0=I$ and $U_1$ is a Fourier transform.
\end{definition}
\begin{proposition}
For the tripartite flower state, $E_{sq}^q(AA':BB':CC')=3+\log d$, $E_{sq}^c(AA':BB':CC')=3+{\frac{3}{2}}\log d$ for the version $I_3$ and $E_{sq}^q(AA':BB':CC')=E_{sq}^c(AA':BB':CC')=3+\log d$ for $S_3$ while $E_{sq}^{q(c)}(A:BB':CC')=0$ for both $I_3$ and $S_3$.
\end{proposition}
{\it Proof.~} The minimization over state extensions $\rho_{AA'BB'CC'Z}$ in squashed entanglement is equivalent to a minimization over CPTP channels $\Lambda$ : $X\rightarrow Z$, acting on the purifying system $X$ for $\rho_{AA'BB'CC'}$, \be \rho_{AA'BB'CC'Z}=(id_{AA'BB'CC'}\otimes\Lambda)\Phi_{AA'BB'CC'X}. \ee Now we use the result proved in \cite{{CW-lock}} that asserts for the bipartite flower state $\sigma_{AA'BB'}$: \be S(AA'Z)+S(BB'Z)-S(AA'BB'Z)-S(Z)\ge 2+\log d,\label{eq:lock} \ee where the equality can be reached by trivial $Z$ (i.e. $Z$ is product with the rest) or by $Z=X$, or finally, by measuring $Z$ in basis $\{|i\>\}$.

For the q-squashed $I_3$,
\begin{eqnarray*}
I(AA':BB':CC'|Z)= AA'Z+BB'Z\nonumber\\
+CC'Z-AA'BB'CC'Z-2Z\nonumber\\
\end{eqnarray*}
\vskip-1cm
\begin{eqnarray}
&=& {1\over 2}\underbrace{[AA'Z+BB'Z-AA'BB'CC'Z-Z]}_{\ge 2+\log d}\label{eq:lock1}\\
&+&{1\over 2}\underbrace{[BB'Z+CC'Z-AA'BB'CC'Z-Z]}_{\ge 2+\log d}\label{eq:lock2}\\
&+&{1\over 2}\underbrace{[CC'Z+AA'Z-AA'BB'CC'Z-Z]}_{\ge 2+\log d}\label{eq:lock3}\\
&+&{1\over 2}\underbrace{[AA'BB'CC'Z-Z]}_{\ge -S(AA'BB'CC')=-\log d}\label{eq:lock4}\\
&\ge& 3+\log d,
\end{eqnarray}
where we omitted S() for brevity. Ineqs. (\ref{eq:lock1}),(\ref{eq:lock2}) and (\ref{eq:lock3}) come from Ineq.(\ref{eq:lock}) and the facts that $S_{\rho}(AA'Z)=S_{\rho}(BB'Z)=S_{\rho}(CC'Z)=S_{\sigma}(AA'Z)=S_{\sigma}(BB'Z)$ and $S_{\rho}(AA'BB'CC'Z)=S_{\sigma}(AA'BB'Z)$, Ineq.(\ref{eq:lock4}) from the subadditivity of von Neumann entropy. It is easy to check that $I(AA':BB':CC'|Z)=3+\log d$ if the extension state is the purification one.

If the register is classical, then $\rho_{AA'BB'CC'Z}$ is necessarily separable in $AA'BB'CC':Z$. Therefore Ineq.(\ref{eq:lock4}) should be replaced by $[S(AA'BB'CC'Z)-S(Z)]\ge 0$. On the other hand, if $\rho_{AA'BB'CC'Z}$ is the state after the measurement on $X$ in the basis $\{|1\>,\cdots,|d\>\}$, then the bound is achieved. So $E_{sq}^{c}=3+{\frac{3}{2}}\log d$.

For the q-squashed $S_3$,
\begin{eqnarray*}
S(AA':BB':CC'|Z)=AA'BB'Z+BB'CC'Z\nonumber\\+CC'AA'Z-2AA'BB'CC'Z-Z\nonumber
\end{eqnarray*}
\vskip-0.5cm
\ben
={1\over 2}\underbrace{[AA'BB'Z+BB'CC'Z-AA'BB'CC'Z-Z]}_{\ge 2+\log d}\label{eq:slock1}\\
+{1\over 2}\underbrace{[BB'CC'Z+CC'AA'Z-AA'BB'CC'Z-Z]}_{\ge 2+\log d}\label{eq:slock2}\\
+{1\over 2}\underbrace{[CC'AA'Z+AA'BB'Z-AA'BB'CC'Z-Z]}_{\ge 2+\log d}\label{eq:slock3}\\
+{1\over 2}\underbrace{[Z-AA'BB'CC'Z]}_{\ge -S(AA'BB'CC')=-\log d}\label{eq:slock4}\\
\ge 3+\log d,
\een
where again we omit S() for brevity. The above inequalities hold for similar reasons. One also can check that $S(AA':BB':CC'|Z)=3+\log d$ if the extension of the state is trivial (i.e. the state itself tensored with some state of $Z$) that also means the c-squashed entanglement is equal to the q-squashed one.

On the other hand, $\rho_{ABB'CC'}$ is separable, so $E_{sq}^{q(c)}(\rho_{ABB'CC'})=0$ for both versions. This shows that the multipartite squashed entanglement is lockable. \blacksquare

The results can be straightforwardly generalized to the $m$-partite flower state. In particular for $I_m$, we have $E_{sq}^q= m+\log d$, and $E_{sq}^c= m+\frac{m}{2}\log d$, while for $S_m$ we have $E_{sq}^q=E_{sq}^c=m+\log d$.

{\bf Remark:} It is notable that during the proof, the q-squashed entanglement of flower state is achieved by the purifying register for $I_3$ that is necessarily quantum (the equality in Ineq.(\ref{eq:lock4}) holds only if $AA'BB'CC'Z$ is in a pure state), while by trivial extension for $S_3$ that can be classical. Moreover the c-squashed entanglement can be arbitrarily larger than the q-squashed one. This gives the negative answer to the open question: `` Is q-squashed entanglement equal to c-squashed entanglement for the version of $I_m$?" for the multipartite case. For the bipartite case, we cannot see it from the flower state. Note that quite recently Brandao \cite{Brandao-PhD} has proved that $E_{sq}^c$ and $E_{sq}^q$ are different even in bipartite case, for antisymmetric Werner states. Now, one can also ask, whether the two versions of squashed entanglement, coming from $S_m$ and $I_m$ respectively, are different from each other. For pure states and the flower states, we see that q-squashed entanglement is equal for two versions. Are these two versions the same in the end? We don't know now. But as we have shown, for c-squashed entanglement the two versions {\it are} different. We then guess, that also for q-squashed entanglement they are different.

\section{Multipartite axiomatic monotones as an upper bound on distillable key}
\label{sec:upper-bound}

In this section we will show, that the quantum squashed entanglement is an upper bound on multipartite distillable
 key \cite{Augusiak-multi}. We will do this by proving that a wide class of entanglement monotones, which satisfy
some natural axioms are upper bounds on distillable key. That is we generalize the results of \cite{EkertCHHOR2006-ABEkey} (see also \cite{MC-phd}) to a multipartite case. We also revisit the analogous result within the so called classical key agreement framework \cite{Maurer_key_agreement} (see also \cite{Wyner_key_agreement,CsisarKorner_key_agreement}), that has been given in \cite{Cerf-secr-mono}. There, an even a more general approach has been developed: the rates of transition
 between the probability distributions under LOPC operations are considered, and the classical multipartite secrecy
monotones are shown to govern these rates. In section \ref{subsec:classical-ubounds} we show again that the classical secrecy monotones are upper bounds on classical distillable key, addressing more carefully the issue of asymptotic
 continuity and consider  more explicitly the presence of the eavesdropper.

In what follows, we will use a short-hand notation for multipartite systems. In particular ${\tilde{A}_1\ldots
 \tilde{A}_m}$ we denote as $\tsystems$. We will assume that $\tilde{A}_i=A_iA_i'$ and sometimes we denote
$\tsystems$ more explicitly as $\systems$. We denote also ${A_1\ldots A_m}$ as $\csystems$ and ${A_1'\ldots A_m'}$ as $\psystems$. Following this rule, the Hilbert space  $\hcal_{\tilde{A}_1}\ot\ldots\ot\hcal_{\tilde{A}_m}$ will be denoted as $\tmhilbert$, and $\hcal_{\tilde{A}_1}^{(n)}\ot\ldots\ot\hcal_{\tilde{A}_m}^{(n)}$ which is isomorphic with
 the tensor product of $n$ copies of $\tmhilbert$, as $\tmnhilbert$. Consequently, the state $\rho$ of such system
will be denoted as $\rho_{(m)\tilde{A}}$. Analogously we mean $\rho_{(m)A'}$ as acting on $B(\pmhilbert) =B(\hcal_{A_1'}\ot\ldots\ot\hcal_{A_m'})$ and $\rho_{(m)A}$ as the state acting on $B(\mhilbert) =B(\hcal_{A_1}\ot\ldots\ot\hcal_{A_m})$.

\subsection{ Multipartite key distillation}

Let us first recall the definition of multipartite distillable key in the context of private quantum states \cite{AH-pditdist}. It was introduced in \cite{AH-pditdist} and is studied in \cite{Augusiak-multi}, as a direct generalization of the bipartite case. Namely, instead of twisted maximally entangled bipartite states, the parties have to distill a twisted $GHZ$ state, so that $m$-partite private dit has a form:

\be \gamma_m^{(d)}=\sum_{i,j=0}^{d-1}{1\over d}|i...i\>\<j...j|\ot U_i\rho_{A_1',...,A_m'}U_j^{\dagger}, \ee where $U_i$ are some unitary operations acting on ${\cal H}_{A_1',...,A_m'}$.

The optimal rate of key $\log d \over n$ distilled by means of LOCC operations from $n$ copies of the input state is the multipartite distillable key, which is denoted as $K_D^{(m)}$. Let us note here, that in what follows, for technical reasons we will denote by $m$ the number of parties of a multipartite state and by $n$ the umber of copies of the state. More formally, we have:

{\definition For any given state
$\initrho\in \bcal(\tmhilbert)$
let us consider a sequence $P_n$ of $LOCC$ operations such that $P_n(\initrho^{\ot n})=\sigma_n$. 
A set of operations ${\cal P} \equiv \cup_{n=1}^{\infty} \{P_n\}$ is called a \pdit\  distillation protocol of state $\initrho$ if there holds \be \lim_{n\rightarrow \infty} ||\sigma_n-\gamma^{d_n}_m|| = 0, \ee where $\gamma_{d_n}$
is a multipartite \pdit\, whose  key  part is of dimension $d_n\times d_n$.

For a protocol $\cal P$, its rate is given by \be {\cal R}({\cal P})=\limsup_{n\rightarrow \infty} {\log d_n \over n}. \ee The distillable key of state $\initrho$ is given by \be K_D^{(m)}(\initrho)=\sup_{\cal P}{\cal R}(\cal P). \ee \label{def:key-rate} }

{\remark \label{rem:linear-com} In Sec. \ref{subsec:linear} we will show that the maximization in the definition of $K_D^{(m)}$ can be restricted to protocols whose communication complexity grows at most linearly in the number of copies of the initial state $\initrhoe$ (lemma \ref{lemma:linear}).  Moreover, the optimal yield can be obtained with the use of a number of systems (needed for implementing local operations), that grows also linearly in the number of the input copies (remark \ref{rem-ancillas}). Hence, if $d=\dim \tmhilberte <\infty$ then the dimension of the output of the protocol is bounded by $\log \dim \tmhilberte \leq c n \log d $, for some constant $c$. }


\subsection{Upper bounds from monotonicity}
\label{subsec:upper-monot} Following the paper by Ekert et al. \cite{EkertCHHOR2006-ABEkey}, we know that a function is an upper bound on distillable key $C_D$, if it satisfies some natural axioms. Their proof can be generalized to the multipartite case, as it does not depend on the number of parties. Actually, the only fact which is needed here is that there is a protocol $P$ with the property that on the $n$ copies of the $m$-partite state $\initrho$, satisfying \be ||\initrho-\gamma^{(d)}_m||\leq \ep \ee that uses only linear communication (in the number of copies of $\rho$) and has a rate close to $\log d$ (i.e. $\log d - f(\ep)$ with $f$ vanishing as $\ep$ approaches zero). More formally it is stated in the theorem below a multipartite version of the  Corollary 3.2. from \cite{EkertCHHOR2006-ABEkey}.

\begin{theorem}\label{thm:mul_monot_bound}
  Let $M(\rho)$ be a function mapping (m)-partite quantum states $\rho
  \equiv \rho_{\systems}$ into the positive numbers such that the following
  holds: \bee
\item Monotonicity: $M(\Lambda(\rho))\leq M(\rho)$ for any LOCC
  operation $\Lambda$.
  \label{cond-1}
\item Asymptotic continuity: for any states $\rho^n,\sigma^n$ acting on
  $\tmnhilbert$, the condition $\|\rho^n
  -\sigma^n\| \to 0$ implies $ {1\over \log r_n}\big|
  M(\rho^n)-M(\sigma^n)\big| \to 0 $ where $r_n=\dim \tmnhilbert$.
  \label{cond-2}
\item \label{cond-3} Normalisation: $ M(\gamma^{(\ell)})\geq\ell \ .$ \eee Then the
  \emph{regularisation} of the function $M$ given by $
  M^\infty(\rho)=\limsup_{n \to \infty} {M(\rho^{\ot n})\over n} $ is
  an upper bound on $K_D^{(m)}$, i.e., $M^\infty(\rho_{\systems}) \geq
  K_D^{(m)}(\rho_{\systems})$ for all $\rho_{\systems}$ with $\dim \tmnhilbert <\infty$.  If in addition $M$ satisfies \bee

\item [4.] Subadditivity on tensor products: $M(\rho^{\otimes n})\leq n M(\rho)$, \eee then $M$ is an upper bound for $K_D^{(m)}$.
\end{theorem}

\begin{proof}
This theorem follows as a corollary from the multipartite version of the theorem 3.1. of \cite{EkertCHHOR2006-ABEkey}. Since the proof of the multipartite version does not differ from the one for the tripartite case, we refer the reader to \cite{EkertCHHOR2006-ABEkey}. We note however, that the latter proof uses the fact stated in lemma A.1. of \cite{EkertCHHOR2006-ABEkey}. We provide in section \ref{subsec:linear} an analogous lemma  in the multipartite case - which proves the first statement of our remark \ref{rem:linear-com}. We show there that the amount of classical communication needed by an optimal key distillation protocol is linear in the number of systems $n$.

We note here also that the proof given in \cite{EkertCHHOR2006-ABEkey} based implicitly on the fact that only communication can increase the dimension of the support of the output state of an optimal key distillation protocol. Although this fact does hold, as it can be proved directly using quantum information techniques, we provide shorter argumentation in Remark \ref{rem-ancillas}.
\end{proof}

{\example (multipartite squashed entanglement as an upper bound on distillable key)

As an example of the application of theorem \ref{thm:mul_monot_bound}, we show here that the normalized multipartite squashed entanglement is an upper bound on the distillable key: \be {1\over m} E^{q}_{sq}(\rho) \geq K_D^{(m)}(\rho). \label{eq:example} \ee

To see this we observe that this monotone easily fulfils the conditions (\ref{cond-1}),(\ref{cond-2}) and (4) of theorem \ref{thm:mul_monot_bound}, as is shown in sections \ref{sec:mon},\ref{sec:cont} and \ref{sec:add} respectively. What remains to check the normalization condition which we show in the observation below.

}

{\Obs For any natural $m\geq 2$ there holds: $E_{sq}^q(\gamma^{(d)}_m) \geq m\log d$. }

{\it Proof.~} The proof is just a direct generalization of that given in \cite{MC-phd} for the bipartite case ($m=2$). In what follows we fix a natural $m\geq 3$ arbitrarily.

We consider a state $\rho_{\systems} = |\Psi_+^{(d)}\>\ot \rho_\psystems$ and its extension to system $E$ $\rho_{\systems E}\equiv \rho$. We have also $\gamma \equiv \gamma^{(d)}_{\systems E} = U\ot\id_E(\rho_{\systems E})U^{\dagger}\ot\id_E$, where $U = \sum_{i=0}^{d-1} |i\ldots i\>\<i\ldots i|\ot U_i^{\psystems}$ is a twisting operation.

We then have the following facts:
\begin{enumerate}
\item \label{fact:first} \ben S({\systems E})_\gamma = S({\systems E})_\rho \nonumber\\= S({\psystems E})_\rho = S({\psystems E})_{\gamma_i}, \een where $\gamma_i^{\psystems E} = U_i^{\psystems}\ot \id_E (\rho_{\psystems E})[U_i^{\psystems}]^{\dagger}\ot \id_E $. \item \label{fact:second} $S(E)_{\gamma_i} = S(E)_{\gamma}$, where $\gamma_E$ ($\gamma_i^E$) denotes the $E$ subsystem of $\gamma_{\systems E}$ ($\gamma^{\psystems E}_i$).

\item \label{fact:third} \be \forall_{i\in 1,\ldots,m} \,\, S(A_iA_i'E)_{\gamma} = \log d + \sum_{k=0}^{d-1} {1\over d} S(A_i'E)_{\gamma_k}. \ee
\end{enumerate}

The first fact is obvious due to unitarity relations between $\rho_{\systems E}$ and $\gamma_{\systems E}$, as well as $\rho_{\psystems}$ and $\gamma_i^{\psystems}$. The second one is just a consequence of the fact that $\gamma_i^E = \gamma_E$. The third one follows from the fact that the $A_iA_i'E$ subsystem is a $cqq$ state and from the joint entropy theorem \cite{Nielsen-Chuang}.

Having established these facts, the argumentation is straightforward. Let $E$ be the fixed extension. By definition of the squashed entanglement \cite{Tucci2002-squashed,Winter-squashed-ent} we have:
\ben
&&E^{q}_{sq}(\gamma_m^{(d)}) = \inf_E I(\systems |E)_{\gamma}\nonumber\\
&&=\inf_E [ \sum_{i=1}^m S(A_iA_i'E)_{\gamma} \nonumber\\
&&~~~~~~~~~~~~~~~- (m-1)S(E)_{\gamma}  -S(\systems E)_{\gamma}]\nonumber\\
&&= m\log d + \inf_E [\sum_{i=1}^m \sum_{k=0}^{d-1}{1\over d} S(A_i'E)_{\gamma_k} \\
&&~~~~~~~~~~~~~~~-(m-1)S(E)_{\gamma} - S({\systems E})_\gamma]\nonumber\\
&&= m\log d +\inf_E [\sum_{k=0}^{d-1}{1\over d}( \sum_{i=1}^m S(A_i'E)_{\gamma_k}\label{eq:secondd}\\
&&~~~~~~~~~~~~~~~-(m-1)S(E)_{\gamma_k} - S({\psystems  E})_{\gamma_k})]\nonumber\\
&&= m \log d + \inf_E \sum_{k=0}^{d-1}{1\over d}I(\psystems |E)_{\gamma_k} \label{eq:third}\\
&&\geq m\log d. \een

In  equality (\ref{eq:secondd}) we have used the fact (\ref{fact:third}) while in (\ref{eq:third}) the facts (\ref{fact:first}) and (\ref{fact:second}). The last inequality is thanks to the positivity of conditional multipartite mutual information. Since the choice of $m\geq 3 $ was arbitrary, we have proved the lemma.

\subsection{Linear communication cost of optimal multipartite key distillation}
\label{subsec:linear} In this section we show, that, as in the bipartite case, an optimal key distillation protocol needs only a linear amount of communication. To this end, we use the following theorem, that is proved in \cite{Augusiak-multi}:

{\theorem  For any state $\initrho$  of the systems $\tsystems$ with $\tilde{A}_i = A_iA_i'$, there is a protocol $\cal P$ with a rate: \be K_D^{\cal P}(\initrho)\geq \min_{i\neq 1} I(A_1:A_i)_{\psi_\rho} - I(A_1:E)_{\psi_\rho}, \ee where $\psi_\rho$ is a purification of the state $\initrho$ to system $E$. On $n$ input copies the protocol $\cal P$ uses $c n$ bits of classical communication, for some constant $c>0$. \label{theorem:main} }

{\remark Another property of the protocol given in \cite{Augusiak-multi}, is that it uses the ancillary systems to implement local operations with dimension that is only linear in the number of input states $n$.  To see this, we just observe that the party which broadcasts to the other parties does not need any ancillary systems, and the remaining $m-1$ parties needs no more of them than the rate of the key, which is clearly linear in $n$. \label{rem-ancillas} }

We can give now the final lemma which proves the statement from remark \ref{rem:linear-com}. This is a modified version of the lemma A.1., given in \cite{EkertCHHOR2006-ABEkey}.

\begin{lemma}
\label{lemma:linear} (Comp. lemma A.1. \cite{EkertCHHOR2006-ABEkey}).

  The maximization in the definition $K_D^{(m)}$ can be restricted to
  protocols that use communication at most linear in the number of
  copies of $\rho_{\systems}$. This implies that if $\dim \tmhilbert =d <\infty$ then $\dim \tmnhilbert \leq c n\log d$, where $c$ is a constant independent of $d$ and $n$.
\end{lemma}
\begin{proof} In what follows we directly reproduce a bipartite proof given in \cite{EkertCHHOR2006-ABEkey}.

Let $\{\Lambda_n\}_{n \in \mathbb{N}}$ be a key distillation protocol (with communication not necessarily linear in $n$) with rate $R$. For any fixed $\epsilon >0$, there exists an $n_0$ such that \ben|| \Lambda_{n_0}(\rho_{\systems}^{\otimes n_0}) - \gamma^{\ell_{n_0}}_{\csystems} ||_1 \leq \epsilon,\een and  $\frac{\ell_{n_0}}{n_0} \geq R-\epsilon$. Consider now key distillation from many copies of $\sigma_{\csystems E}:=\Delta (\Lambda_{n_0}(\rho_{\systems}^{\otimes n_0}))$, where $\Delta$ is a measurement of Alice and Bob in their computational bases followed by tracing out the systems $\psystems$ (and the same for Bob). Using Fannes' inequality \cite{Fannes} we can bound the difference in the mutual informations by \ben \min_{i\neq 1} I(A_1:A_i)_{\psi_\sigma} -I(A_1:E)_{\psi_\sigma} \geq \ell_{n_0} - 4 \epsilon \ell_{n_0}  -6 \eta(\epsilon),\een where $\psi_\sigma$ is a purification of the state $\sigma$.
 As it is assured by theorem \ref{theorem:main}, Alice and Bob are now able to achieve a rate $\min_{i\neq 1} I(A_1:A_i)_{\psi_\sigma} -I(A_1:E)_{\psi_\sigma}$ using
communication linear in the number of copies of $\sigma_{\csystems}$,
We therefore have modified the protocol $\{\Lambda_n\}_{n \in \mathbb{N}}$ achieving a rate $R$ into a protocol $\{\tilde\Lambda_n\}_{n \in \mathbb{N}}$ with a rate \ben\tilde R\geq (1-c \epsilon) (R-\epsilon) - \frac{c' h(\epsilon)}{n_0}.\een The amount of communication needed is proportional to the number of copies of $\rho_{\systems}$. Since $\epsilon$ was arbitrary we obtain a sequence of protocols (with communication linear in the number of copies $\rho_{\systems}$) approaching the rate $R$.
\end{proof}

\subsection{Analogous results for multipartite classical key agreement}
\label{subsec:classical-ubounds}

In this section we will consider the analogous result to the ones given in theorem \ref{thm:mul_monot_bound}, in the realm of {\it classical key agreement} \cite{Maurer_key_agreement,Wyner_key_agreement,CsisarKorner_key_agreement}, that has been found in \cite{Cerf-secr-mono}. Within this paradigm, the distillable key from an $m+1$-partite distribution is the amount of key that can be obtained in the asymptotic limit from the $n$ realizations of a distribution via local operations and communication between $m$ parties $\csystems$ which is public i.e. listened to by the $m+1$-th party - Eve, per copy of a distribution. We will show that a function of an $m+1$-partite distribution, which satisfies axioms of monotonicity, continuity, normalization and subadditivity, is an upper bound on the amount of secure key that can be distilled from the $m+1$ partite distribution ({\it distillable key}).

In what follows $P_{(m)A}$ will denote a distribution defined on $\chi_{(m)A} = \chi_{A_1}\times\ldots\times\chi_{A_m}$. The $n$ such distributions resides consequently on $\chi_{A_1}^{n}\times\ldots\times\chi_{A_m}^n$, which is denoted as $\chi_{(m)A}^{n}$. Analogously, $\chi_{A_1'}^{n}\times\ldots\times\chi_{A_m'}^{n}\times E'$ will be denoted as $\mnchiep$.

For the sake of completeness of the presentation, we begin with definition of an LOPC protocol:

{\definition An LOPC protocol $\cal P$ is an family $\{\Lambda\}_n$ of classical channels $\Lambda_n:(\mnchie) \rightarrow \mnchiep$, which are a finite number of concatenation of local operations (local channels) and public communication steps (communicating some number to one honest party and the eavesdropper). }

We pass now to definition of distillable secure key, that is the key that can be obtained from a given distribution via an (optimal) LOPC protocol:

{\definition We say that LOPC protocol $\cal P$ is a classical key distillation protocol for a distribution $\mdist \in \mchi$, if \be \lim_{n\rightarrow \infty} ||\Lambda_n(\mdist^{\ot n}) - K_{\csystems}^{l_n}|| =0, \ee where $K_{\csystems}^{l_n} = {1\over l_n}(\ind)\delta_{\ind}\ot P_E$ is the ideal key distribution on $\mchi$, for some distribution $P_E$ of the eavesdropper. The rate of a protocol is given by \be \rcal(\pcal) = \limsup_{n\rightarrow \infty} {\log l_n \over n}. \ee Then the classical distillable key of a distribution $P_\csystems$ is defined as supremum of the rates \be C_D^{(m)}(P_\csystems) = \sup_{\pcal}\rcal(\pcal). \ee }

Consequently, we have the following theorem:

{\theorem (See also a general result of \cite{Cerf-secr-mono}) \label{thm:classical-version} Let $M$ be a function mapping $m+1$-partite classical distributions (defined on $\chi_m^n = \chi_{{\tilde{A}}_1}^{n}\times\ldots\times{\chi_{{\tilde{A}}_m}^{n}}\times{\chi_{E}^{n}}$) into positive numbers such that there holds:

\bee \item Monotonicity: $M(\Lambda(P))\leq M(P)$ for any LOPC
  operation $\Lambda$.
\item Asymptotic continuity: for any distributions $P^n,Q^n$ defined on
  $\chi_m^n$, the condition $\|P^n
  -Q^n\| \to 0$ implies $ {1\over \log r_n}\big|
  M(P^n)-M(Q^n)\big| \to 0 $ where $r_n=\dim \chi_m^n$ and $\tilde{A}_i = A_iA_i'$.
\item 
Normalisation: $ M(K_\csystems^{(\ell)})\geq\ell \ .$ \eee Then the
  \emph{regularisation} of the function $M$ given by $
  M^\infty(P)=\limsup_{n \to \infty} {M(P^{\ot n})\over n} $ is
  an upper bound on $C_D^{(m)}$, i.e., $M^\infty(P_{\systems}) \geq
  C_D^{(m)}(P_{\systems})$ for all $P_{\systems}$ with $\dim \chi_m^n <\infty$.  If in addition $M$ satisfies \bee

\item [4.] Subadditivity on tensor products: $M(P^{\otimes n})\leq n M(P)$, \eee then $M$ is an upper bound for $C_D^{(m)}$. }

{\proof The proof of this theorem is in full analogy to that of  the theorem \ref{thm:mul_monot_bound} hence we do not present it here. It needs however the facts analogous to that from remark \ref{rem:linear-com} which we present below.

We first argue, that the maximization in the definition of the distillable key $C_D^{(m)}$ can be taken over the protocols which have only linear amount of communication. To prove this, we observe, that the quantum multipartite Devetak-Winter (D-W) protocol provided in \cite{Augusiak-multi} can be applied also in the classical case, with appropriate change of quantum operations to classical data processing operations. Such modified quantum multipartite D-W protocol we will call a {\it classical multipartite D-W protocol}.

The quantum multipartite D-W protocol, as a multipartite extension of the Devetak-Winter protocol, works for multipartite ccq states. It is easy to see that is essentially classical i.e. can be done with help of local classical operations and public communication. Indeed, instead of quantum POVM's that are used by the honest parties in this protocol \cite{DevetakWinter-hash,Augusiak-multi}, one can consider classical operations, that are based on the classical Slepian-Wolf theorem \cite{CoverThomas}. Moreover, if some amount of key can be obtained from the multipartite ccq state, then it can be also obtained from this state after projecting onto Eve's subsystem turning it into a classical state. For this reason, the classical multipartie D-W protocol applied to the probability distribution gains at least that much of secure key, as the quantum one. Since the quantum multipartite D-W protocol uses only a linear amount of communication, so will  the classical one (see lemma \ref{lemma:linear}).

We argue now, that the dimension of the ancillary system used to implement local operations in the classical multipartite D-W protocol is also linear in the number of input copies. This holds by the same property of the quantum D-W protocol, that has been noted in remark \ref{rem-ancillas}.

In summary, we have that the dimension of the output variable of the honest parties is only linear in the number of input copies. This ends the proof of the above theorem. }

{\remark Note that from \cite{EkertCHHOR2006-ABEkey} it follows that the above theorem holds for a special class of distributions. These are distributions for which Eve's distribution is a function of Alice's and Bob's joint distribution.
Unfortunately, one can not follow this approach to extend the results to all distributions. This is because there are distributions for which $C_D \neq K_D$. For this reason we can not use directly the ``unified formalism'' given in \cite{EkertCHHOR2006-ABEkey} to get the result. }

{\example (multipartite intrinsic information as an upper bound on classical distillable key) We will show now, that there holds: \be {1\over (m-1)}I(\csystems \downarrow E)_P \geq C_D^{(m)}(P), \ee where $P$ is some distribution of variables $\csystems$. As in case of the multipartite squashed entanglement, we merely check that this function satisfies the axioms of the above theorem. In what follows, we will not assume, that the multipartite intrinsic information is realized  on a particular variable $E$, which seems however to be true due to it being true for the bipartite intrinsic information (see \cite{Christandl-Ren-Wolf-intrinsic}).}

{\it Asymptotic continuity.-} It is noted in \cite{Cerf-secr-mono}, that the multipartite intrinsic information is simply continuous in its argument (a weaker condition than that of asymptotic continuity).
The asymptotic continuity of this function follows from general result contained in proposition 2 of Ref. \cite{Synak05-asym} as explained in the bipartite case in section V.B of \cite{Synak05-asym} (see also \cite{renner-wolf-gap}). The reason is the same as in the case of mixed convex roof: the operation of "arrowing" preserves asymptotic continuity, and the continuity of conditional multipartite information follows from the Fannes inequality \cite{Nielsen-Chuang}.


{\it Monotonicity.-}The monotonicity of the multipartite intrinsic information for the case where Eve is in a product state with the honest parties is shown in \cite{Cerf-secr-mono}. It is also noted there that in the general case the monotonicity also holds and the proof is a direct generalization of the monotonicity proof they have provided, hence we just briefly sketch the idea, for completeness of presentation. To show monotonicity under local operations, it remains to rewrite the multipartite
 intrinsic information in terms of bipartite conditional mutual information and use the chain rule for the latter.
Then, the data processing inequality assures that the function can not increase. To see that communication can not increase it either, it is enough to see that it can be rewritten as if only  Eve had the communication. It
 is then clear that discarding such information can only increase the multipartite intrinsic information.

{\it Normalization.-}Since for the multipartite intrinsic information there is a representation given in Eq. (\ref{eq:multi-bi-cond}), it is easy to see that on the ideal distribution $K_{\csystems E}^{\log d }$ it takes value: \be I(K_{\csystems}^{\log d }|K_E^{\log d}) = \sum_{i=1}^m S(A_i) - S(\csystems) = (m - 1)\log d, \ee where $d$ is the dimension of the support of $A_1$. Hence, by theorem \ref{thm:classical-version}, we get that ${1\over (m-1)} I(\csystems \downarrow E) \geq C_D^{(m)}$.

{\example A similar result is true also for another monotone which is a version of $S_m$ extended to $m+1$ systems which we denote following \cite{Cerf-secr-mono}  as $S_m\downarrow$.
\ben
&&S_m\downarrow := \inf_{E\rightarrow \bare} [
I(A_1:A_2\ldots A_m|\bare)\nonumber \\
&&+ I(A_2:A_3\ldots A_m |A_1 \bare) +I(A_3:A_4\ldots A_m|A_2A_1\bare)\nonumber \\
&&+\ldots + I(A_{m-1}:A_m|A_{m-2}\ldots A_1\bare)]. \label{eq:multi-rule}. \een
To get rid of the infimum, we realize $S_m\downarrow$ via the sequence of channels $E\rightarrow \bare_n$. Then for each $n$ we observe that each term as in the above sum is just a difference of mutual informations.
 E.g. $I(A_2:A_3\ldots A_m|\bare_n) = I(A_2:A_3\ldots A_m \bare_n) - I(A_2:\bare_n A_1)$, that forms a `pair'.
It is then easy to see, that apart from the first such 'pair' which equals $\log$ of the dimension of the support of $A_1$, all the others are equal to zero. Hence we have that \be
 S_m\downarrow \geq C_D^{(m)}.
\ee }

\section{Conclusions}
\label{sum} In summary, we present two new multipartite entanglement measures, the generalized q-squashed entanglement and the c-squashed one.  There were obtained by squashing the conditional multipartite mutual information over the extension states (i.e. taking the infimum).
When the ancilla register is quantum we get the q-squashed entanglement,  and when it is classical we get the c-squashed one. Each measure has two versions corresponding to two possible multipartite mutual informations $I$ and $\Sn$. The q-squashed entanglements are additive, while for the c-squashed additivity is an open question.

Notice that we consider two types of correlation between the system and the ancilla in the extension states -- completely quantum correlation and classical correlation. Suppose now that  we allow only the extensions of the form $\sum p_i\rho^{i}_{\aaa}\otimes\phi_{E}^{i}$ where $\phi_{E}^{i}$ is an unnecessarily orthogonal states. That is we restrict the extension to separable states between the system and the ancilla. Is it still a good entanglement measure by squashing? The answer is yes. As a matter of fact, one can check that the proof of monotonicity of $\eqsq$ can be applied to any convex extension set. For example one can allow the extension to be $PPT$  state between the system and the ancilla. We thus have a whole family of entanglement measures.

We have shown also, that any multipartite entanglement measure which satisfies reasonable axioms, is an upper bound on the multipartite distillable key $K_D$. Using this fact we proved that q-squashed entanglement is an upper bound on $K_D$. It is not difficult to see that the same holds for the regularized mixed convex roof of the function $I_c$. We also revisit the analogous result in the classical realm given in \cite{Cerf-secr-mono},
 addressing more carefully the issue of asymptotic continuity,
as well as the presence of the eavesdropper.

{\bf Note added:} After completing this manuscript we have noticed a paper by G.A. Paz-Silva and J.H. Reina, quant-ph/0702177, which also proposed a procedure to build  entanglement measures based on mixed convex roof.
 Also, we have got to know, that Patrick Hayden and Ivan Savov
have independently proposed a multipartite version of squashed entanglement.

{\bf Acknowledgements}: The authors are grateful to the two anonymous referees for their detailed comments and suggestions that helped to improve the paper. M.H. would like to thank Piotr Badzi{\c a}g for discussions. We gratefully thank the Newton Institute, Quantum Information Science 2004 where much of this work was completed.


\bibliographystyle{IEEEtran}

\begin{biographynophoto}{Dong Yang}
received his Ph.D. in Physics from Zhejiang University in 2002. He
joined the Laboratory of Quantum Information at China Juliang
University in June 2007. Currently his research interests include
quantum information theory, entanglement theory, and quantum
computation.
\end{biographynophoto}
\begin{biographynophoto}{Karol Horodecki}
 graduated from University of Gda\'nsk, Gda\'nsk, Poland, in
2004, where he is currently teaching assistant at the Institute of Informatics.
having received the Ph.D. degree in computer science from the University of Warsaw in 2009.
He is a coauthor of several papers on processing of quantum information. His
research interests are in quantum information theory.
\end{biographynophoto}
\begin{biographynophoto}{Micha\l{} Horodecki}
graduated from University of Gda\'nsk, Gda\'nsk Poland, in
1997 and received the Ph.D. degree in physics from the same university in 2000.
In 2007, he was appointed Professor at the University of Gda\'nsk. He is a
coauthor of Quantum Information (New York: Springer-Verlag, 2001). His
main achievements include pioneering research on the entanglement of mixed
states, in particular, bound entanglement and quantum state merging. His
research interests are in quantum information theory and the foundations of
quantum physics.
\end{biographynophoto}
\begin{biographynophoto}{Pawe\l{} Horodecki} graduated from University of Gda\'nsk, Gda\'nsk, Poland, in
1995 and received the Ph.D. degree in physics from the Technical University of
Gda\'nsk, Gda\'nsk, Poland, in 1999.
In 2008, he became a Professor at the Technical University of Gda\'nsk. He
is a coauthor of Quantum Information (New York: Springer-Verlag, 2001). His
main achievements include pioneering research on the entanglement of mixed
states, in particular, bound entanglement. His research interests are in quantum
information theory and foundations of quantum physics.
\end{biographynophoto}
\begin{biographynophoto}{Jonathan Oppenheim}
graduated from the University of Toronto, Toronto, ON,
Canada, in 1993, and received the Ph.D. degree under B. Unruh from the University
of British Columbia, Vancouver, BC, Canada, in 2001.
Currently, he is Royal Society University Research Fellow at the Department
of Applied Mathematics and Theoretical Physics (DAMTP), University
of Cambridge, Cambridge, U.K. His research interests include quantum information
theory, foundations of quantum mechanics, quantum gravity, and black
hole thermodynamics.
\end{biographynophoto}
\begin{biographynophoto}{Wei Song}
received his Ph.D. in Physics from the University of Science and
Technology of China in 2008. Currently he works as a Lecturer in
the School of Physics and Telecommunication Engineering of the
South China Normal University. His research interests include
various aspects of quantum information science and quantum
simulation in condensed matter physics.
\end{biographynophoto}
\end{document}